\documentclass[prd,aps,amsfonts,eqsecnum,superscriptaddress,nofootinbib,longbibliography,notitlepage]{revtex4-1}

\usepackage{graphicx}
\usepackage{xcolor}
\usepackage{caption}
\usepackage{rotating}
\usepackage{amsmath,amssymb,graphics,amsthm,isomath}
\usepackage{subfigure}

\usepackage[colorlinks=true, urlcolor=violet, linkcolor=blue, citecolor=red, hyperindex=true, linktocpage=true]{hyperref}
\usepackage[capitalise,compress]{cleveref}

\newtheorem{thm}{Theorem}
\numberwithin{thm}{section}

\newtheorem{prop}[thm]{Proposition}

\makeatletter
\renewcommand{\p@subsection}{}
\renewcommand{\p@subsubsection}{}
\makeatother

\usepackage{xcolor}
\usepackage{mathtools}

\newcommand{\LL}{\mathcal{L}}
\newcommand{\QQ}{\mathbb Q}

\newcommand{\pfrac}[3]{\left(\frac{#1}{#2}\right)^{#3}}

\usepackage{dsfont}

\begin{document}

\title{Operator growth bounds in a cartoon matrix model}

\author{Andrew Lucas}
\email{andrew.j.lucas@colorado.edu}
\affiliation{Department of Physics and Center for Theory of Quantum Matter, University of Colorado, Boulder CO 80309, USA}

\author{Andrew Osborne}
\email{andrew.osborne-1@colorado.edu}
\affiliation{Department of Physics and Center for Theory of Quantum Matter, University of Colorado, Boulder CO 80309, USA}

\begin{abstract}
We study operator growth in a model of $N(N-1)/2$ interacting Majorana fermions, which live on the edges of a complete graph of $N$ vertices. Terms in the Hamiltonian are proportional to the product of $q$ fermions which live on the edges of cycles of length $q$.  This model is a cartoon ``matrix model": the interaction graph mimics that of a single-trace matrix model, which can be holographically dual to quantum gravity.  We prove (non-perturbatively in $1/N$, and without averaging over any ensemble) that the scrambling time of this model is at least of order $\log N$, consistent with the fast scrambling conjecture.  We comment on apparent similarities and differences between operator growth in our ``matrix model" and in the melonic models.
\end{abstract}

\date{\today}

\maketitle
\tableofcontents

\section{Introduction}
Our best hint to the unification of gravity with quantum mechanics arises from the holographic principle \cite{Susskind:1994vu,Maldacena:1997re}.  Some quantum mechanical systems in $D$ spacetime dimensions, with a large $N$ number of degrees of freedom ``per site" (e.g. in a lattice discretization), have been conjectured to describe non-perturbative quantum gravity in $D+1$ dimensions.  An obvious question then arises:  \emph{which} quantum systems do this? 

It has been realized that there is a simple check for whether or not a quantum system might realize quantum gravity holographically.  Roughly speaking, the time it takes to ``scramble" quantum information \cite{Sekino:2008he} scales as $\log N$ in a holographic model.  A computable and specific definition of scrambling is the growth in a suitably defined out-of-time-ordered correlation function (OTOC) \cite{Shenker:2013pqa}: schematically, \begin{equation}
  -  \langle [A(t),B]^2\rangle  \sim \frac{1}{N}\mathrm{e}^{\lambda t}.
\end{equation}
The ``fast scrambling" behavior of holographic models is assured by the finite Lyapunov exponent $\lambda>0$.  A classic theorem in mathematical physics \cite{Lieb1972} then requires that $N$ degrees of freedom interact in a spatially non-local way in order to realize gravity.

Luckily, there are many ways to realize non-local interactions among $N$ degrees of freedom experimentally, including in cavity quantum electrodynamics \cite{Leroux_2010,thompson2020} or trapped ion crystals \cite{Britton_2012}.  There has been a significant amount of recent work to try and understand whether it is possible to realize quantum holography in a quantum system with non-local interactions \cite{Lashkari_2013,Barbon:2012zv,Chew_2017,Chen_2018,Bentsen_2019, Bentsen_2019prl, Bentsen_2019prx, Alavirad_2019,Lewis_Swan_2019,li2020fast,belyansky2020minimal,Yin:2020pjd}. One microscopic model which genuinely realizes exponential operator growth is the Sachdev-Ye-Kitaev (SYK) model \cite{Sachdev:1992fk, Sachdev:2015efa,Maldacena:2016hyu,Kitaev:2017awl}, which shares a universal low energy effective theory \cite{Almheiri:2014cka,Maldacena:2016upp} with two-dimensional Jackiw-Teitelboim gravity \cite{Jackiw:1984je,Teitelboim:1983ux}.  

While there is not an obvious route to realizing the SYK model in a near-term experiment, its analysis is still instructive. At infinite temperature, OTOCs in the SYK model have a particularly simple interpretation as counting the average size of a growing operator \cite{Roberts:2018mnp}, as we will explain in detail later.  In a certain large $N$ limit, a simple form was found for the time-dependent operator size distribution of a growing operator in the SYK model \cite{Roberts:2018mnp}.  It was later proved \cite{Lucas:2019cxr} that the qualitative picture of operator growth in the large $N$ limit, found in \cite{Roberts:2018mnp}, in fact constrains the growth of operators at large but finite $N$ in the SYK model, hence leading to a mathematically rigorous proof of the fast scrambling conjecture in a ``gravitational" model.  Since aspects of operator growth might be accessible in experiments using multiple quantum coherences \cite{Garttner_2017}, it is worthwhile to understand whether the pattern of operator growth is robust from one holographic model to another.

The purpose of this paper is to prove that the same bounds that constrain operator growth in the SYK model also apply to a cartoon single-trace matrix model.  Matrix models are intricately related to quantum gravity: the original conjecture of holography relates the $\mathcal{N}=4$ super-Yang-Mills theory (whose degrees of freedom are $N\times N$ matrices) to quantum gravity in an asymptotically $\mathrm{AdS}_5\times \mathrm{S}^5$ spacetime \cite{Maldacena:1997re}.  Matrix models have also been argued to underlie string theory and M-theory more generally \cite{Banks:1996vh,Ishibashi:1996xs}. Previous work studying chaos in matrix models includes \cite{Stanford:2015owe,Grozdanov:2018atb}.   Note that we are \emph{not} talking about ``random matrix theory", which is also is intricately connected to quantum gravity, albeit in a somewhat different way \cite{Weingarten:1982mg,Kazakov:1985ea,Douglas:1989ve,DiFrancesco:1993cyw,Saad:2019lba}. 

From our perspective, matrix models are interesting because their microscopic Hamiltonians are very different, a priori, from the SYK model.  Relative to the SYK model, the number of terms in the Hamiltonian is parametrically smaller, yet which terms arise in the Hamiltonian are much more constrained.  The $1/N$ expansion of the matrix models is related to the a certain genus\footnote{Diagrams of order $1/N^g$ can only be embedded in a ``planar" way on a two-dimensional Riemann surface of genus $\ge g$.  In this paper, a rather different and more abstract notion of graph genus will control the $1/N$ expansion.} of the corresponding Feynman diagrams \cite{tHooft:1973alw,Brezin:1977sv}, and is qualitatively different from the $1/N$ expansion in the SYK model and other melonic models.  Lastly, matrix models are \emph{not random}, while at least in the SYK model, randomness is essential in order to realize a chaotic fast scrambler (at least rigorously \cite{Lucas:2019cxr}).  Despite these microscopic differences, we will prove that the matrix models are also fast scramblers, and that the operator size distribution is constrained in exactly the same way as it was in the SYK model.  Our result suggests that there may be some universality to operator growth in holographic models. Moreover, our proof also gives an illustrative microscopic example of how a microscopic model without randomness can nonetheless appear to be just as ``frustrated" as the random SYK model; this frustration, together with a large $N$ limit, is crucial in realizing exponential OTOC growth.  We expect that our framework could be generalized to prove operator growth bounds for melonic models without randomness \cite{Gurau:2011xp,Gurau:2016lzk,Witten:2016iux,Klebanov:2016xxf,Gubser:2018yec}.

One important question which we leave unresolved is about finite temperature chaos.  The Lyapunov exponent $\lambda$ is believed to be universal at low temperatures in every holographic model \cite{Maldacena:2015waa}.   It is debated \cite{Susskind:2018tei,Brown:2018kvn,Lucas:2018wsc,Qi:2018bje} how to understand or interpret this universality in the language of operator growth.  A rigorous resolution of this issue, perhaps including a rigorous proof of the conjecture of \cite{Maldacena:2015waa} (which rests on believable but physical assumptions about correlators in chaotic systems), is a challenging question beyond the scope of this paper.

\section{Summary of results}
We consider a model of $N(N-1)/2$ Majorana fermions $\psi_{ij} = -\psi_{ji}$, where $1\le i,j \le N$ are positive integers.  Intuitively speaking, these fermions live on the undirected edges of a (complete) graph where all $N$ vertices are connected.  The Hamiltonians we consider are schematically of the form \begin{equation}
    H = \sum_{i_1,\ldots, i_q=1}^N J_{i_1\cdots i_q} \psi_{i_1i_2}\psi_{i_2i_3}\cdots \psi_{i_qi_1}. \label{eq:fermiontrace}
\end{equation}
where \begin{equation}
    |J_{i_1\cdots i_q}| \le \frac{1}{N^{(q-2)/2}}
\end{equation}
are coefficients which are not necessarily random.  The coefficients $J_{i_1\cdots i_q}$ are \emph{not} necessarily random.  We are interested in studying particular realizations of this model.  

Why is this a ``cartoon matrix model"?   Let us compare to a more conventional matrix model (in zero space and one time dimension) with bosonic $N\times N$ matrix degrees of freedom: \begin{equation}
    H = \mathrm{tr}\left[ \Pi^2 + m^2 \Phi^2 + \lambda \Phi^q  \right],
\end{equation}
where component-wise $[\Phi_{ij},\Pi_{ij}] = \mathrm{i}$; we also assume $\Phi_{ij}=\Phi_{ji}$ and $\Pi_{ij}=\Pi_{ji}$.  We have only allowed single-trace terms in the Hamiltonian. Writing the interaction term in component form, we find \begin{equation}
    \mathrm{tr}\left[  \Phi^q  \right] = \sum_{i_1,\ldots, i_q=1}^N \Phi_{i_1i_2}\Phi_{i_2i_3}\cdots\Phi_{i_qi_1}. \label{eq:bosontrace}
\end{equation}
The interaction structure of (\ref{eq:fermiontrace}) is analogous to (\ref{eq:bosontrace}), which is why we call it a ``cartoon matrix model".   We will see that the allowed interactions in (\ref{eq:fermiontrace}) are quite constraining and have non-trivial consequences on operator growth.

%[when you define the model, call the coefficient $1/N^{(q-2)/2}$ $\sigma$.  The reason is that I also suggest we put some further $q$ dependence in this coefficient $\sigma$, which we'll determine after finishing section 4, so that $\lambda$ is $q$-independent at large $q$.]

Returning to our model (\ref{eq:fermiontrace}), we study the growth of a ``typical" OTOC at infinite temperature:
\begin{equation}
    C_{ij}(t) = \frac{2}{N(N-1)} \sum_{k<l} \frac{\mathrm{tr}(\lbrace \psi_{ij}(t),\psi_{kl}\rbrace^2)}{\mathrm{tr}(1)}.
\end{equation}
As described in \cite{Roberts:2018mnp}, we can interpret the right hand side of the above equation as an ``average operator size", as measured in an operator size distribution (which we define precisely in the next section).  By studying a (non-Markovian) stochastic process governing the evolution of this size distribution, which we obtain from the many-body Schr\"odinger equation, we can prove that there exists a scrambling time \begin{equation}
    t_{\mathrm{s}} = \kappa \log N, \label{eq:fastscrambling}
\end{equation}
where $\kappa >0$ is a constant which is finite in the $N\rightarrow \infty$ limit, such that \begin{equation}
    C_{ij}(t) \le \frac{c}{N^2} \mathrm{e}^{\lambda_{\mathrm{L}}t}, \;\;\; (0<|t|<t_{\mathrm{s}})
\end{equation}
for some finite $N$-independent constant $c$.  $t_{\mathrm{s}}$ is referred to as the ``scrambling time" for operator growth, and (\ref{eq:fastscrambling}) is one statement of the fast scrambling conjecture \cite{Sekino:2008he}.  We prove that in the matrix models (\ref{eq:fermiontrace}), the exponent \begin{equation}
    \lambda_{\mathrm{L}} < 2\sqrt{2}(q-2)^2. \label{eq:lambdaLboundloose}
\end{equation}
We do not expect this bound on the Lyapunov exponent is tight -- even the $q$ dependence above may not be tight.  Nevertheless, for sensible values of $q$ (e.g. $q=4$), (\ref{eq:lambdaLboundloose}) is sufficient to provide a mathematical proof of the fast scrambling conjecture in the cartoon matrix model.

Although the canonical matrix models of string theory include bosonic degrees of freedom (and thus our results do not immediately apply), we conjecture that many of our qualitative results remain relevant for these bosonic models.  As we will detail later, there are many qualitative similarities (though at least one important difference) between operator growth in our cartoon matrix model and the melonic models such as the SYK model.   It seems plausible that every holographic model must have qualitatively similar operator growth to these models, wherein operators grow in the fastest (and most quantum coherent) way possible, in contrast with random unitary circuit models of operator growth \cite{Bentsen_2019}.

\section{Mathematical preliminaries}

\subsection{Sets and graphs}
We begin by fixing $N\in 4\mathbb{Z}^+$ and constructing the complete (undirected)
graph on  $N$ vertices $\mathrm{K}_N$. Define 
\begin{equation}\label{eqn:defV}
    V = \{1,2,\dots,N\}.
\end{equation}
We define an (undirected) \textbf{graph} as a tuple $(E_G,V_G)$, where the \textbf{edge set} $E_G$ is a set of two element sets drawn from $V$, and the \textbf{vertex set} \begin{equation}
    V_G := \lbrace v \in V : \; v\in e\text{ for some } e\in E_G\rbrace.
\end{equation}
We will often denote $e\in E_G$ with $e\in G$ and $v\in V_G$ with $v\in G$ when clear from context.  Note that for the complete graph $\mathrm{K}_N$, the edge set is given by \begin{equation}
    E_{\mathrm{K}_N} := \lbrace \lbrace i,j\rbrace : 1\le i<j\le N\rbrace,
\end{equation}
while the vertex set is given by $V_{\mathrm{K}_N}=V$.  We will henceforth always refer to an undirected graph $G$ as simply a ``graph", as no directed graphs arise in this paper.
Because we view every graph as a subgraph of $\mathrm{K}_N$ (for some $N$), there is 
a unique edge between any two vertices in all of the following discussion.

We write the \textbf{genus}
%\footnote{In some communities this quantity is called \textit{excess} or \textit{surplus}}  andy: not sure we need this
of a graph $G$ as 
$g_G$ and define it so that 
\begin{equation}\label{eqn:genusdefn}
 g_G := 1 + |E_G| - |V_G|.
\end{equation}
%\AO{might not need this; going to look for dependencies}
%For a graph $G$ with multiple connected components, we write the number of
%connected components as $h_G$ and the genus of the  $i$'th connected component
%as $g_G^i$ where  $1 \leq i \leq h_G$.
%Note that 
%\begin{equation}\label{eqn:altgenusdefn}
%   1 + |E_G| - |V_G| = 1 - h_G + \sum_{i=1}^{h_G} g_G^i.
%\end{equation}

We define a \textbf{cycle} $C=(e_1,e_2,\ldots,e_l)$ to be an ordered list of distinct edges (i.e. no edge can appear twice) with the property that $e_i = (v_i,v_{i+1 \text{ (mod }l)})$.  Such a cycle is said to have length $q$.   We define 
\begin{equation}\label{eqn:admissible}
    \mathcal{C}_{l}^N = \{ C \subset \mathrm{K}_N\,:\, C \text{ is a cycle of length $l$}\}.
\end{equation}

Define, for an arbitrary graph $X$, and arbitrary vertex $v$, the \textbf{degree}
of $v$ in $X$ as 
\begin{equation}\label{eqn:degdefn}
  \deg_X v = |\{e\in E_X\,:\, v\in e  \}|.
\end{equation}

In what follows, we will be most interested in cycles of a fixed length $q\in2\mathbb{Z}$, corresponding to the number of fermions in each term in our Hamiltonian, as in (\ref{eq:fermiontrace}).  A cycle from $\mathcal{C}_q^N$ is called \textbf{admissible}. 
We will suppress the $N$ superscript on $\mathcal{C}_q$ when context allows.

For two sets $A$ and $B$ define the \textbf{set difference} of $A$ and $B$ as 
\begin{equation}\label{eqn:setminus}
  A\setminus B := \{x \in A\,:\, x\not\in B \}.
\end{equation}
It need not be the case that $B \subset A$ for $A\setminus B$ to be well--defined.
Let us now define the \textbf{symmetric difference}, written as 
\begin{equation}\label{eqn:symmdiff}
  A \triangle B := (A\cup B)\setminus (A\cap B) :=
  \{x \in A\cup B\,:\, x \not\in A\cap B\}.
\end{equation}
Let $X$ be a set operation from $\{\cap,\cup,\triangle,\setminus\,\}$.
We can use  $X$ and two graphs  $G_1$ and $G_2$ to build a third graph
$G_1 X G_2$ in the following way:
let $E_{G_1 X G_2} = E_{G_1} X E_{G_2}$ and let 
\begin{equation}\label{eqn:opdefn}
  V_{G_1 X G_2} := \{ v \in V\,:\, \text{ there exists } e \in E_{G_1 X G_2} \text{ so that } v \in e\}.
\end{equation}
The graphs constructed in this way in general respect the properties of 
the set operations on edge sets, but not on vertex sets. 
For example, it is always the case that 
\begin{equation}\label{eqn:setminusexample} 
  E_{G_1 \setminus G_2}\cap E_{G_2} = \emptyset
\end{equation}
but it is possible that 
\begin{equation}\label{eqn:setminusvertex}
  V_{G_1\setminus G_2}\cap V_{G_2} \neq \emptyset.
\end{equation}
However, since every edge in  $G_1$ is in one of $E_{G_1\cap G_2}$ or $E_{G_1\setminus G_2}$, we must have that 
\begin{equation}
  V_{G_1} = V_{G_1\cap G_2}\cup V_{G_1 \setminus G_2}.
\end{equation}

\subsection{Majorana fermions and operator size}

For some fixed $N \in 4\mathbb{Z}^+$, define 
\begin{equation}\label{eqn:defH}
    \mathcal{H} = (\mathbb{C}^2)^{\otimes \frac{N(N-1)}{4}}.
\end{equation}
On each edge $e\in\mathrm{K}_N$ we define a Majorana fermion 
$\psi_{e} \in \mathcal{B}:=\text{End}(\mathcal H)$.
Each fermion operator is Hermitan.   The fermions obey the anticommutation relation 
\begin{equation}\label{eqn:anticomm}
    \{\psi_{e_1} , \psi_{e_2}\} = 2 \mathbb{I}[e_1 =e_2].
\end{equation}
where $\mathbb{I}[\cdots]$ is the indicator function, which is 1 if its argument is true and 0 otherwise. 
We further associate with each subgraph $G$ of $\mathrm{K}_N$ an operator $\psi_G$ (not necessarily Hermitian):
\begin{equation}
  \psi_G = \prod_{ e \in E_G} \psi_{e} \label{eq:subops}
\end{equation}
where the order of the product is prescribed\footnote{However, it will not be important to us to give a precise prescription.} so as to fix the sign 
of $\psi_G$. 
 We introduce the shorthand $|G)$ to denote $\psi_G$:  the notation is deliberately reminiscent of the bra-ket notation for linear algebra, since we can naturally turn $\mathcal{B}$ into an inner product space: for $A,B\in \mathcal B$, define the inner product
\begin{equation}\label{eqn:prod}
    (A|B):= \frac{1}{2^{\frac{N(N-1)}{4}}}\text{tr}(A^\dagger B).
\end{equation}
Let $\| \cdot\|_2$ be the Frobenius ($\mathrm{L}^2$) norm on $\mathcal{B}$ induced by (\ref{eqn:prod}). 
For an arbitrary $\mathcal M \in \text{End}(\mathcal B)$, define 
\begin{equation}\label{eqn:opnorm}
  \| \mathcal M \| = \sup_{G \in \mathcal B} \frac{\| \mathcal{M} |G)\|_2}{\||G)\|_2} = \sum_{G,G^` \in \mathcal B}
  \frac{|(G| \mathcal M |G')|}{\| |G) \|_2 \| |G^`)  \|_2}
\end{equation}
$\lVert\cdot\rVert$ above denotes the conventional operator norm on $\mathrm{End}(\mathcal{B})$.  

\subsection{The cartoon matrix model}\label{sec:model}
We now formally introduce our cartoon matrix model. Fix some $q \in 2\mathbb{Z}^+$ with $q \leq N$. Define 
% factor of i^{q/2} makes this self--adjoint in general if each fermion is self--adjoint
\begin{equation}\label{eqn:hamilt}
  H :=  \mathrm{i}^{\frac{q}{2}}\sum_{C\in \mathcal{C}_q^N} J_C \psi_C := \sum_{C\in \mathcal{C}_q^N} H_C
\end{equation}
where $J: \mathcal{C}_q^N \rightarrow [-\sigma,\sigma]$ are real numbers, with \begin{equation}\label{eqn:sigma}
    \sigma = \frac{1}{N^{\frac{q-2}{2}}} %*** fix here later if necessary ***%
.\end{equation}
We define time evolution on $\mathcal B$ by the group of one--parameter 
automorphisms generated by $H$ in (\ref{eqn:hamilt}) which will be written as  $\LL$.
That is, we define  $\LL\in\text{End}(\mathcal{B})$ so that 
\begin{equation}\label{eqn:louis}
    \LL = \mathrm{i}[H,\cdot].
\end{equation}
For some subgraph $G\subseteq \mathrm{K}_N$, we write  \begin{equation}
    |G(t)) = \mathrm{e}^{\LL t}|G).
\end{equation}
We also define
\begin{equation}\label{eqn:partl}
    \LL_C = \mathrm{i}[H_C,\cdot].
\end{equation}
We immediately see that \begin{equation}
    \mathcal{L} = \sum_{C\in\mathcal{C}^N_q} \mathcal{L}_C.
\end{equation}

\begin{prop}\label{prop:fermioncommute}
    Let $G$ be a subgraph of $\mathrm{K}_N$ and $C \in \mathcal{C}_q^N$.   Then
    \begin{equation}\label{eqn:symmdiff}
        \LL_C |G) = 2  \alpha J_C \mathbb{I}[|E_{C\cap G}|\in 2\mathbb{Z}+1]|G\triangle C).
    \end{equation}
    with $\alpha \in \mathbb{C}$ obeying $|\alpha|=1$.  
    %\andy{there should be no i here; I removed it in equation above.  in the SYK paper I introduced factors of $\mathrm{i}$ so that all $\psi_G$ were Hermitian, you probably have to do this as well here to get this identity to work out.}
   % \AO{The point of this is to justify that I can do counting with a symmetric difference. I don't really care about the prefactor, do I? I don't think I understand what you're getting at. Are you suggesting that $\psi_G = i^{x} \prod_{e\in G}\psi_e$ for some x?} \andy{fair point.  I guess we don't really need to worry about these factors of i; however I would then suggest just replacing $\alpha$ with a complex number?  the thing is, we do know that since $H$ is Hermitian, whatever factor of i went into $H$ will make }
\end{prop}
\begin{proof}
    Let $G$ be a subgraph of $\mathrm{K}_N$ and let $e$ be an edge in $\mathrm{K}_N$. 
    Then it is clear from (\ref{eqn:anticomm}) that, if $e\not\in E_G$
    \begin{equation}\label{eqn:outgraph}
        \psi_G \psi_e = (-1)^{|E_G|}\psi_e \psi_G
    \end{equation}
    and if $e \in E_G$, then 
    \begin{equation}\label{eqn:ingraph}
        \psi_G \psi_e = (-1)^{|E_G| - 1} \psi_e \psi_G.
    \end{equation}
    From (\ref{eqn:ingraph}) and (\ref{eqn:outgraph}) it is clear that 
    \begin{equation} \label{eq:psiGpsiC}
        \psi_G \psi_C = (-1)^{|E_G||E_C| - |E_{C\cap G}|} \psi_C \psi_G.
    \end{equation}
    By construction $|E_C|$ is even so we see that 
    $[\psi_G,\psi_C]$ is nonzero if and only if $|E_{C\cap G}|$ is odd.
    If this is the case, then we have that 
    \begin{equation}\label{eqn:commid}
        [\psi_G,\psi_C] = 2 \psi_G \psi_C 
    \end{equation}
     and by repeated use of (\ref{eqn:anticomm}) it can be seen that 
     \begin{equation}\label{eqn:prodsymm}
        \psi_G \psi_C = \beta \psi_{G\triangle C}
    \end{equation}
    for some $\beta \in \{-1,1\}$. Adding now the fact that  
    $\LL_C=\mathrm{i}^{\frac{q}{2}+1} J_C [\psi_C,\cdot] $, the proposition is established.
\end{proof}
\begin{prop}\label{prop:span}
    The space $\mathcal B$ is spanned by elements in the set 
    \begin{equation}\label{eqn:spanset}
        \{|G) \,:\, G \subseteq \mathrm{K}_N\}
    \end{equation}
\end{prop}
\begin{proof}
(\ref{eq:subops}) is a natural isomorphism between the basis vectors of $\mathcal{B}$ (i.e. operators on $\mathcal{H}$) and all possible subgraphs (which need not be connected) of $\mathrm{K}_N$.  In other words, there is an isomorphism between elements of $\mathbb{Z}_2^{E_{\mathrm{K}_N}}$ (subgraphs) and basis vectors $|G)$.   Note that when $G=\emptyset$, the corresponding operator is the identity.

It remains to show that the basis vectors are orthogonal.  It is clear that, for some graph $G$, $(G|G) = 1$ from 
  (\ref{eqn:anticomm}) and (\ref{eqn:prod}). Let $G$ be a graph with at least a single edge so that 
  $|E_G|$ is odd. Take some edge $e$ so that  $e\not\in E_G$
  Then by (\ref{eqn:outgraph}) and (\ref{eqn:ingraph}) 
  \begin{equation}\label{eqn:zerotrnin}
    0=\text{tr}(\psi_e \{\psi_G ,\psi_e\}) = 2\text{tr}(\psi_G)
  \end{equation}
  by the cyclic property of traces.
  If $|E_G|$ is even, then choose some edge  $e\in E_G$ and let  $E_G = E_{G_0} \cup \{e\}$.
  Then,  by (\ref{eqn:ingraph}) we have that 
  \begin{equation}
    \{\psi_{G_0},\psi_e \} = 0.
  \end{equation}
  From this it follows that 
  \begin{equation}\label{eqn:treven}
    0 = \text{tr}(\{\psi_{G_0},\psi_e \}) = 2 \alpha \text{tr}(\psi_G)
  \end{equation}
  for some $\alpha \in\mathbb{C}$ with $|\alpha|=1$.
  Hence, by (\ref{eqn:treven}) and (\ref{eqn:zerotrnin}), every nontrivial product of fermions is traceless.
  From (\ref{eqn:prodsymm}), we can see then that 
  \begin{equation}\label{eqn:prodzero}
    (G | G^\prime) = \mathbb{I}[G=G^\prime]
  \end{equation}
  which confirms orthogonality.
\end{proof}

\begin{prop} \label{prop:parity}
    For $v\in V$, let \begin{equation}
     \Psi_v := \prod_{e\in E_{\mathrm{K}_N}: v\in e}  \psi_e.
    \end{equation}
    Then if $N\in4\mathbb{Z}$ and $H$ is given by (\ref{eqn:hamilt}), for all $v,u\in V$:
    \begin{equation}
        [\Psi_v, H]=[\Psi_u,\Psi_v]=0. \label{eq:Psicommute}
    \end{equation}
\end{prop}
\begin{proof}
Let $C\in \mathcal{C}^N_q$.  Then if vertex $v\in C$, $\deg_C(v)\in 2\mathbb{Z}$.  Hence the fermion product $\psi_C$ has an even number of fermions in common with every $\Psi_X$.  Invoking Proposition \ref{prop:fermioncommute} proves that $[H,\Psi_v]=0$.

To prove that $[\Psi_u,\Psi_v]=0$, observe that $\Psi_v$ is a product of $N-1$ fermions.  $N-1$ is odd.  $\Psi_v$ and $\Psi_u$ share exactly one fermion corresponding to edge $\lbrace u,v\rbrace$, so according to (\ref{eq:psiGpsiC}), $\Psi_v\Psi_u = \Psi_u\Psi_v$. 
\end{proof}

Note that $\Psi_v$ is a Hermitian operator and that $\Psi_v^2=1$.  Crudely speaking, we might expect that (\ref{eq:Psicommute}) leads to degeneracies in the spectrum of $H$, where typical eigenvalues have a degeneracy of order $\exp[N]$.  Keeping in mind that the total number of states in the Hilbert space is $\exp[N^2]$, however, we expect that such degeneracies are rather mild.  Regardless, the purpose of this paper is to understand the growth of operators, so we will not study in detail the eigenspectrum of $H$.

\subsection{Time evolution as a quantum walk}
% review stuff along the lines of section 3 of the SYK paper.
% not sure whether to put this here or in the previous subsection.
% 
For $0\leq s \leq N$ define the projector  $\QQ_s$ so that, for some subgraph
$G$ of $\mathrm{K}_N$,
\begin{equation}\label{eqn:defproj}
  \QQ_s |G) = \mathbb{I}[|E_G| = s]|G).
\end{equation}
Note that by Proposition \ref{prop:span}, we have for a graph $G$ that 
\begin{equation}\label{eqn:projsum}
  \sum_{s = 1}^N \QQ_s |G(t) ) = |G(t)).
\end{equation}
Thus we may define 
\begin{equation}\label{eqn:prob}
  P_s(G,t) = \frac{(G(t)| \QQ_s |G(t))}{(G(t)|G(t))}
\end{equation}
which is a well--defined probability measure. We say that 
$|G(t))$ is size $s $ with probability $P_s(G,t)$. 
We define the partition of $V$
\begin{equation}\label{eqn:partition}
  R_l = 
  \begin{cases}
    \{1\} & l = 0 \\ 
    \{m\in \mathbb{Z} \,:\, (l-1)(q-2)+1 < m \leq l(q-2)+1\} & 0 < l < N' \\
    \{m\in\mathbb{Z}\,:\, (N'-1)(q-2) + 1 \leq m \leq N   \} & l = N'
  \end{cases}
\end{equation}
with 
\begin{equation}\label{eqn:nprime}
 N' = \left\lceil \frac{N-1}{q-2} \right\rceil.
\end{equation}
Using this we define 
\begin{equation}\label{eqn:blockprob}
  P_l(G,t) = \sum_{s\in R_l}P_s(G,t)
\end{equation}
and 
\begin{equation}\label{eqn:blockproj}
  \QQ_l = \sum_{s\in R_l} \QQ_s.
\end{equation}
We say that $G$ is in block $l$ with probability $P_l(G,t)$ at time $t$. 
Loosely speaking, this partition is defined so that
if $\QQ_{l_0}|G) = 1$ for some $l_0$ then $\QQ_l \LL_C |G) = 0 $ unless $|l-l_0| \leq 1 $ for an arbitrary admissible cycle $C$. 
It is this key observation that leads us to define a quantum walk in the spirit of \cite{Lucas:2019cxr}.
% In order to proceed with this goal we must define
As a notational convenience in the following discussion, we will suppress the 
appearance of $G$ in $P_s(t)$ and  $P_l(t)$ and write 
 \begin{equation}\label{eqn:wavefunc}
   P_s(t) = \varphi_s(t)^2 \text{  and  } P_l(t) = \varphi_l(t)^2.
\end{equation}

With the following proposition from \cite{Lucas:2019cxr}, we make the connection to a quantum walk explicit.
\begin{prop}\label{prop:kdefn}
  Let $H$ be the Hamiltonian defined in (\ref{eqn:hamilt}) with some appropriately fixed  $N$ and $q$, and let $0\leq s, s'\leq N$.  
  Finally, let 
  \begin{equation}\label{eqn:curlyk}
    \mathcal{K}_{s's} = \| \QQ_{s'} \mathcal{L} \QQ_s \|.
  \end{equation}
  Then there are functions $K_{s's}: \mathbb{R} \rightarrow [-\mathcal{K}_{ss'},\mathcal{K}_{ss'}]$ 
  so that 
  \begin{equation}\label{eqn:swalk}
    \frac{\mathrm{d}}{\mathrm{d}t}\varphi_{s}(t) = \sum_{s' < s}K_{ss'}(t)\varphi_{s'}(t) - \sum_{s > s'} K_{s's}(t) \varphi_{s'}(t).
  \end{equation}
  Further if 
  \begin{equation}\label{eqn:walksuppos}
    \mathcal K _l = \max\left( 
      \max_{s\in R_l} \sum_{s'\in R_{l+1}}\mathcal{K}_{s's}, 
      \max_{s'\in R_{l+1}} \sum_{s \in R_l} \mathcal{K}_{s's}
      \right),
  \end{equation}
  then there exist functions $K_l: \mathbb{R} \rightarrow [-\mathcal{K}_l,\mathcal{K}_l]$ so that 
  \begin{equation}\label{eqn:auxham}
    \frac{\mathrm{d}}{\mathrm{d}t}\varphi_l(t) = K_{l-1}(t) \varphi_{l-1}(t) - K_{l}(t)\varphi_{l+1}(t)
  \end{equation}
  provided that $K_{-1}(t) = K_{N'}(t) = 0$. 
\end{prop}
\begin{proof}
  Let $G$ be a subgraph of $\mathrm{K}_N$. For $0 < s \leq N$, let  $|\mathcal{G}_s)$
  be an operator of unit norm so that 
  \begin{equation}
    \QQ_s |G(t)) = \sqrt{P_s(t)}|\mathcal{G}_s(t)) 
  \end{equation}
  which must be unique when $P_s(t) \neq 0$. We henceforth suppress the time dependence in all of $|G)$,  $|\mathcal{G}_s)$, and  $P_s$.
  Then, from (\ref{eqn:louis}) and (\ref{eqn:prob}) we have that 
  \begin{equation}\label{eqn:pdot}
    \frac{\mathrm{d}}{\mathrm{d}t}P_s = (G|[\QQ_s,\mathcal{L}]|G).
  \end{equation}
  Moreover, by (\ref{eqn:projsum}), we may rewrite (\ref{eqn:pdot}) as 
  \begin{equation}\label{eqn:pdotsum}
    \frac{\mathrm{d}}{\mathrm{d}t}P_s = \sqrt{P_s}(\mathcal G_s| \mathcal L| G) - \sqrt{P_s}(G|\mathcal L| \mathcal G_s) = 
    \sqrt{P_s}\sum_{s'} \sqrt{P_{s'}} (\mathcal G _s|\mathcal L | \mathcal G_{s'})-
    (\mathcal G_{s'}| \mathcal L |\mathcal G_s).
  \end{equation}
  Next we will define 
  \begin{equation}\label{eqn:defK}
    K_{ss'} := (\mathcal G_{s}| \QQ_s \mathcal L \QQ_{s'}|\mathcal G_{s'}) =
    (\mathcal G_{s}|\mathcal L|\mathcal G_{s'}).
  \end{equation}
  Observing that $\frac{\mathrm{d}}{\mathrm{d}t}(\varphi_s)^2 = 2 \varphi_s \frac{\mathrm{d}}{\mathrm{d}t}\varphi_s$,
  we directly acquire (\ref{eqn:swalk}). 
  
  The required bound on $K_{ss'}$ is trivial by (\ref{eqn:curlyk}).
  The analogue for block size is derived in exactly the same manner, but (\ref{eqn:walksuppos}) must be demonstrated. 
  Define 
  \begin{equation}\label{eqn:lkdefn}
    K_l = (G|\QQ_{l+1}\mathcal L \QQ_l|G) \leq \|\QQ_{l+1} \mathcal L \QQ_l\|
    := \mathcal{K}_l.
  \end{equation}
  By (\ref{eqn:opnorm}), we see that,   for arbitrary $|O)$ and $|O^`) \in \mathcal B$,
  \begin{equation}\label{eqn:walksuppospf}
  \begin{split}
    \mathcal{K}_l = \sup_{|O),|O^`) \in \mathcal B} \frac{(O|\QQ_{l+1} \LL \QQ_l  |O^`)}{\sqrt{(O|O) (O^`|O^`)}} &
    \leq \sup_{|O),|O^`) \in \mathcal B}\sum_{s \in R_l}\sum_{s' \in R_{l+1}}
    \sqrt{P_s(O,t) P_{s'}(O^`,t)} \|\QQ_s \mathcal L \QQ_{s'}\| \\ 
    \leq \sup_{|O),|O^`)}\sum_{s \in R_l}&\sum_{s' \in R_{l+1}} \frac{1}{2}
    (P_{s}(O,t)+P_{s'}(O',t)) \mathcal{K}_{s's}
    \end{split} .
  \end{equation}
(\ref{eqn:walksuppos}) clearly bounds the right most term of (\ref{eqn:walksuppospf}) as $P_s$ is a well-defined probability distribution.
\end{proof}

\subsection{Bound on the Lyapunov exponent}
Qualitatively, by bounding $\mathcal{K}_l$ from (\ref{eqn:walksuppos}), we can ensure that the system has a finite Lyapunov exponent. 
Formally, we have the following theorem from \cite{Lucas:2019cxr}, quoted without proof.
\begin{thm}\label{thm:sykthm}
If there exists $M \in \mathbb{Z}^+$ and $c > 0 $ so that 
\begin{equation}
   \mathcal{K}_l \leq c(l + 1)  \label{eq:thm35c}
\end{equation}
whenever $l\leq M$ then the Lyapunov exponent has the following property
for each $\alpha > 0 $:
\begin{equation}
    \lambda (t) = \frac{1}{t}\log\left( \sum_{l=0}^{N'}l \,P_l(t) \right) \leq 2c(1+\alpha) \label{eq:thm35alpha}
\end{equation}
provided that 
\begin{equation}\label{eqn:smallt}
|t| < \frac{1}{4c(1+e)}\left( \log M - 2 - \log\log \frac{{N'}^3}{2 \alpha} \right).
\end{equation}
\end{thm}
The main result of this paper is that (\ref{eq:thm35c}) holds for the cartoon matrix model.

\section{Operator growth in the Majorana matrix model}
\subsection{Bounding operator growth rates}
We now state our main result.
\begin{thm}\label{thm:main}
  Let $\epsilon >0$. 
    For the model introduced in Section \ref{sec:model}, 
    \begin{equation}\label{eqn:toprove}
         \mathcal{K}_l \leq \sqrt{2} \frac{8e^\epsilon}{(4 - \epsilon^2)(2-q\epsilon)} (q-2)^2 (l+1)
    \end{equation}
    when 
    \begin{subequations}\label{eq:thmbounds}\begin{align}
        l &< \left\lfloor    \frac{\epsilon \sqrt N}{2 q^2} \right\rfloor, \\
        q\epsilon &\leq 2.
    \end{align}\end{subequations}
\end{thm}
\begin{proof}
Where possible, the method of this proof follows \cite{Lucas:2019cxr}.  We begin with a simple observation:
\begin{prop}
  Define 
  \begin{equation}\label{eqn:Msps}
    M_{s's} = \QQ_s \mathcal L^\intercal \QQ_{s'} \mathcal L \QQ_s.
  \end{equation}
  $M_{s's}$ is symmetric and positive semidefinite. 
  If the maximal right eigenvalue of $M_{s's}$ is $\mu_{s's}$, then 
   \begin{equation}\label{eqn:}
     \mathcal{K}_{s's} = \sqrt{\mu_{s's}}.
  \end{equation}
\end{prop}
\begin{proof}
  Let $|O) \in \mathcal B$ be of unit norm. Then we define 
  \begin{equation}\
    |O') = \QQ_{s'} \LL \QQ_{s} |O)
  \end{equation}
  and we see that 
  \begin{equation}
    \mathcal{K}_{s's} = \sup_{|O')}\sqrt{(O'|O')}.
  \end{equation}
  It follows that 
  \begin{equation}\label{eqn:}
    \mathcal{K}_{s's}^2 = \sup_{|O) \in \mathcal B}(O|M_{s's}|O) = \mu_{s's}
  \end{equation}
  because $M_{s's}$ is symmetric. 
\end{proof}

%From Proposition \ref{prop:indbound}, we seek to bound $(G'|M_{s's}|G)$ for arbitrary $G,G'$ by counting the number of graphs which are of the form $G \triangle C$ for some admissible cycle  $C$.
For each $s \in \mathbb{Z}^+$, define
\begin{equation}\label{eqn:}
  \mathrm{K}_N^s = \{G \subset \mathrm{K}_N \,:\, |E_{G}| = s\}.
\end{equation}

\begin{prop}\label{prop:CW}
  Define  
  \begin{equation}\label{eqn:ansatz}
    |\phi) = \sum_{G' \in \mathrm{K}_N^s} N^{\frac{g_{G'}}{2}} |G') .
  \end{equation}
  Then 
  \begin{equation}\label{eqn:intbound}
    \mu_{s's} \leq \frac{(\phi | M_{s's}|G)}{(\phi|G)}.
  \end{equation}
\end{prop}
\begin{proof}
  Let $\epsilon > 0$ and define 
  \begin{equation}\label{eqn:sgraphs}
    \mathcal{E} = \epsilon \sum_{G,G' \in \mathrm{K}_N^s} |G)(G'|
  \end{equation}
  and
  \begin{equation}\label{eqn:curlMdefn}
    \mathcal{M}_{s's} = M_{s's} + \mathcal{E}_{s^\prime s}.
  \end{equation}
  Denote the maximal eigenvalue of $\mathcal{M}_{s's} $ as  $\nu_{s's}$
  and the maximal eigenvalue of $\mathcal{E}_{s's}$ as $\epsilon_{s's}$. 
  We see that $M_{s's}$ is positive semidefinite and, by construction,  $\mathcal{E}$ is positive definite, so  $\mathcal{M}_{s's}$ is positive definite.  Moreover since $\mathcal{E}$ is irreducible (every entry is strictly positive), so is $\mathcal{M}_{s's}$.  Hence, we may apply the Perron--Frobenius theorem 
  and therefore the Collatz--Weilandt bound \cite{Meyer:2000bk} to bound $\nu_{s's}$:
  \begin{equation}\label{eqn:sep}
  \nu_{s's} \le   \frac{(\phi|\mathcal{M}_{s's}|G)}{(\phi|G)} = 
    \frac{(\phi|M_{s's}|G)}{(\phi|G)} + \frac{(\phi|\mathcal{E}|G)}{(\phi|G)} .
  \end{equation}
The second term above can be crudely bounded:
  \begin{equation}\label{eqn:smplfyE}
    \frac{(\phi|\mathcal{E}|G)}{{(\phi|G)}} = N^{\frac{-g_G}{2}} \epsilon \sum_{G' \in \mathrm{K}_N^s} N^{\frac{g_{G'}}{2}} < \epsilon \times N^{2s} \times N^{s/2}
  \end{equation}
  where we have used the fact that $|\mathrm{K}_N^s| < N^{2s}$ and $g_G < s$.
  Therefore,
  \begin{equation}
      \nu_{s's} \le \frac{(\phi|\mathcal{M}_{s's}|G)}{(\phi|G)} + \epsilon N^{5s/2}.
  \end{equation}
  
  By the triangle inequality, we also know that   \begin{equation}\label{eqn:ineq}
    \mu_{s's} - \epsilon_{s's} \leq \nu_{s's},
  \end{equation}
  and since $\mathcal{E}$ is a rank-1 matrix, it is easy to find its maximal eigenvalue: \begin{equation}
      \epsilon_{s's} = \epsilon |\mathrm{K}_N^s|  < \epsilon N^{2s}.
  \end{equation}
  Therefore, we find that \begin{equation}
      \mu_{s's} < \epsilon N^{2s}\left(1+N^{s/2}\right) + \frac{(\phi|\mathcal{M}_{s's}|G)}{(\phi|G)}.
  \end{equation}
At any finite $N$, we may now take the limit $\epsilon \rightarrow 0$.  Hence we obtain (\ref{eqn:intbound}). 
\end{proof}

In order to bound $\mu_{s's}$, we need only to bound 
\begin{equation}\label{eqn:reworksum}
  \frac{(\phi|M_{s's}|G)}{(\phi|G)} = \sum_{G' \in \mathrm{K}_N^s} N^{\frac{g_{G'} - g_G}{2}} (G'|M_{s's}|G).
\end{equation}
Introducing the shorthand
\begin{equation}\label{eqn:deltag}
  \mathrm{\Delta} g = g_{G'} - g_G
\end{equation}
(which implicitly depends on $G'$), we rewrite (\ref{eqn:reworksum}) as 
\begin{equation}\label{eqn:canbound}
 \frac{(\phi|M_{s's}|G)}{(\phi|G)} =  \sum_{\mathrm{\Delta} g} N^{\frac{\mathrm{\Delta} g}{2}}\sum_{G'\in \mathrm{K}_N^s\,:\,g_{G'} = g_G + \mathrm{\Delta} g}
  (G'|M_{s's}|G).
\end{equation}
It is this inner sum that we will bound by a sufficiently careful counting. Indeed, for every $G' \in \mathrm{K}_N^s$ 
so that  $(G'|M_{s's}|G) \neq 0$, there is a pair of admissible cycles
$C$ and $C'$ so that 
\begin{equation}\label{eqn:nonzero}
  (G'| \QQ_s \LL_{C'} \QQ_{s'} \LL_C \QQ_s|G) \ne 0. 
\end{equation}
Bounding $(G^\prime| M_{s's}|G)$ reduces to counting the number of cycles $C$ and $C^\prime$ which can lead to a given  $\mathrm{\Delta} g$.

Define a \textbf{segment} to be an ordered list of unique edges, so that consecutive
edges share a vertex.  We denote $\mathcal{S}^E$ to be \begin{equation}\label{eqn:segmentdefn}
  \mathcal{S}^E = (e_1,e_2,e_3,\dots,e_k) \text{ with } |e_i\cap e_{i+1}| = 1\text{ for } 1\leq i < k, \text{ and } e_i \ne e_j \text{ if } i\ne j.
\end{equation}
Equivalently, we could specify an ordered list of vertices so that consecutive vertices are connected by unique edges.  We may also list the vertices as 
\begin{equation}\label{eqn:altsegmentdefn}
  \mathcal{S}^V = (v_1, v_2,v_3 ,\dots , v_{k+1}), \text{ where } e_i = (v_i, v_{i+1}).
\end{equation}
We use the superscript $E$ or $V$ to be clear about which perspective is taken.
We say that some segment $\mathcal{S}$ is in some graph  $G$ if every edge (every pair of vertices) is in $E_G$. 
Define a \textbf{path} to be a sequence of segments, $\{\mathcal{S}_i\}_{i=1}^\infty$, so that the last vertex appearing in  $\mathcal{S}_i^V$ is the first vertex appearing in $\mathcal{S}_{i+1}^V$. 

The notion of segments is useful because we can think of breaking up an admissible cycle $C$ into different segments: $C=(\mathcal{S}_1,\mathcal{S}_2,\ldots)$ which alternate between overlapping and not overlapping with $G$: \begin{equation}
    \mathcal{S}_k \subset G \text{ if } k \text{ is odd, } \; \; \mathcal{S}_k \cap G = \emptyset \text{ if } k \text{ is even}.
\end{equation}  
For each admissible cycle $C$,  there is clearly a unique sequence of segments with this property.  We define $\eta(C,G)$ as the total number of segments defined in this way.  Note that $\eta(C,G)\in 2\mathbb{Z}$, and that the choice of $(\mathcal{S}_1,\ldots)$ is ambiguous because of which of the $\frac{1}{2}\eta$ possible segments is chosen to be $\mathcal{S}_1$ (and whether we go around $C$ in one order or the reverse order). 
%For instance, in Fig. \ref{fig:segment}, $\eta(C,G)= 4$. 
It is also clear that 
\begin{equation}\label{eqn:etabound}
  h_{C\setminus G} \leq \frac{1}{2}\eta(C,G) \text{ and } h_{C\cap G} \leq \frac{1}{2}\eta(C,G).
\end{equation}
In our explicit combinatoric bound on $(\phi|M_{s^\prime s}|G)$, we will prefer to count the number of ways to arrange segments, instead of counting cycles $C$ directly.  Of course, since the number of distinct sequences of segments corresponding to a given $C$ is simply $\eta(C,G)$, after accounting for this ``overcounting" we can choose to count segments instead of $C$s.
%The key observation is that, for each path where $\mathcal{S}_i \subset G$ for  $i\leq \eta(C,G)$ odd and $\mathcal{S}_i \subset \mathrm{K}_N\setminus G$ for  $i$ even, 
%there is one admissible cycle and for each admissible cycle there is at least one such path.
%By counting paths, in this way, we can bound the number of cycles obeying certian relationships.
\begin{figure}[t]
    \centering
    \includegraphics[scale=0.8]{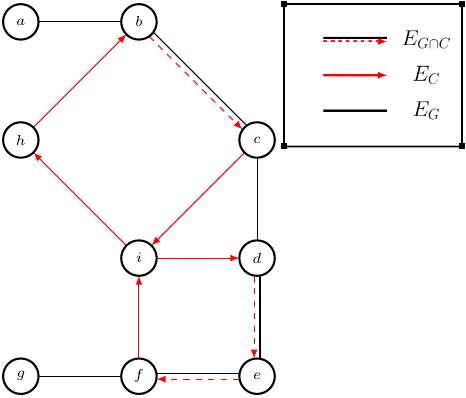}
    \caption{An example of an admissible cycle $C$ and graph $G$ with $\eta(C,G)=4$.  One possible path describing this overlap is $\mathcal{S}_1 = (b,c)$, $\mathcal{S}_2 = (c,i,d)$, $\mathcal{S}_3 = (d,e,f)$, $\mathcal{S}_4 = (f,i,h,b)$.
     }
    \label{fig:segment}
\end{figure}

\begin{prop}\label{prop:formalbound}
  Let $G\subset\mathrm{K}_N$ obey $|E_G| = s$.  Define \begin{equation}
  p = |E_{C\cap G}| = \frac{q - s^\prime + s}{2}.
  \end{equation}
  Choose some $\epsilon \in (0,1)$ and suppose that \begin{subequations}\label{eq:propbounds}\begin{align}
     \epsilon \sqrt{N} &> 2q(\max(s,s^\prime)+q), \\
     \sqrt{N} &> q^2(s+q).
  \end{align}\end{subequations}
  Then we can bound the number of admissible cycles $C$ such that $g_{C\triangle G} - g_G = \delta g$ as  \begin{align} \label{eq:Gammabound}
      \Gamma_{\delta g} &= | \lbrace C \in \mathcal{C}_q^N : C \triangle G \in \mathrm{K}_N^{s^\prime}, g_{G\triangle C} = g_G +\delta g\rbrace |  \leq  
      \frac{4 e^\epsilon}{4 - \epsilon^2}\times s N^{q -p-1}\pfrac{q^2 (s +q )}{\sqrt N}{|\delta g|} \times N^{-\frac{1}{2}\delta g}.
  \end{align} 
\end{prop}
\begin{proof}
  Schematically, we will construct a general admissible cycle $C$ by counting
  the number of choices to be made when constructing a path locally. 
  For this purpose, we prefer to think of choosing vertices over choosing edges. As such, we tend to think of segments in the form of (\ref{eqn:altsegmentdefn}), but we will also use (\ref{eqn:segmentdefn}) when it is convenient to do so. %  We can, with no loss of generality, take the first segment in a path to be in $G$, and the next in $\mathrm{K}_N\setminus G$ and so on. Further, the last segment in our path must be in  $\mathrm{K}_N\setminus G$. 
  %\AO{This is much better than what I had before. I see now what you meant about not needing to care explicitly about any of the $\beta,\dots$ that I was playing with before. This is tighter because you didn't have to play any games with approximating sums.. or rather the inclusion of $y$ dealt with that for you. Also, you did something  with the factors of two that I don't really understand. I needed to include various factors of two to deal with the freedom to choose which direction to proceed after starting an odd--numbered segment. Provided that the one other comment below is resolved, I need to recalculate a bunch of the proceeding discussion; I will start that in the morning. }

 We proceed by the method of generating functions.  Let us define a function $\Gamma(y)$ such that \begin{equation}
    \Gamma(y) = \sum_{\delta g=-q}^q \Gamma_{\delta g} y^{\delta g}.
\end{equation}
To construct $\Gamma(y)$, let us first count how many cycles there are for a fixed value of $\eta(C,G)$, as well as fixed segment lengths $|\mathcal{S}_i^E| = L_i$.  We will sum over the possible values of $L_i$ at the end.

We start with the first segment $\mathcal{S}_1$.  There are $s$ choices of first edge $e$ to choose, and for each given edge $e=(u,v)$, no more than two choices for the ordering of $\mathcal{S}_1^V$:  $(u,v,\ldots)$ or $(v,u,\ldots)$.  Now, let us imagine defining the graph $\tilde G_1 = G\setminus \lbrace e\rbrace$.  Algorithmically, $\tilde G_{i-1} \rightarrow \tilde G_{i}$ will be updated as we march along the path $C$, deleting edges on odd segments $\mathcal{S}_{2k-1}$ and adding them on even segments $\mathcal{S}_{2k}$.  Note that $\tilde G_q = G\triangle C$.  So, after this first step, $\tilde G_1$ corresponds to the graph $G$ with our first edge $e$ removed.  Let us now bound the function $\tilde\Gamma^{1}_1$, corresponding to how many choices we made \emph{thus far}: \begin{equation}
    2\tilde\Gamma^{1}_1 < 2s . \label{eq:2Gamma1}
\end{equation}
The left hand side has a factor of 2 (which we will carry throughout the computation) since we could clearly read a cycle in either direction;  it is easiest to just deal with this double counting at the end.  (\ref{eq:2Gamma1}) is an inequality because if, for example, we pick an edge $e$ which connects to a degree 1 vertex, there is (if $p>1$) only one way to orient the cycle:  the first vertex in $\mathcal{S}_1^V$ must have degree 1.

Note that $g_{\tilde G_1}$ and $g_G$ are, in general, not the same.  However, we will \emph{not} include a factor of $y$ at this stage.  Without loss of generality, suppose $\mathcal{S}_1^V= (v_1,v_2,\ldots)$.  Suppose that $v\in \tilde G_1$;  then at the very last step of the path $C$, our edge will necessarily add 1 to $\delta g$ since it will connect two vertices already in $\tilde G_{q-1}$.  This $+1$ will undo the fact that $g_{\tilde G_1} = g_G - 1$.  Now suppose that $v\notin \tilde G_1$; then the final edge in the path $C$ will not change the genus: $\tilde G_{q-1}=\tilde G_q$.  In either case, we are free to ignore the genus change at both the first and the last step, which always cancel.  We will do so as it is convenient.

Let us now move on the remaining steps in $\mathcal{S}_1^V$, assuming that $L_1>1$.  If $\deg_{\tilde G_1}(v_2) = 1$, then when we delete the next edge, $\tilde G_2$ will have one fewer vertex ($v_2$) and one fewer edge than $\tilde G_1$;  hence $g_{\tilde G_1} = g_{\tilde G_2}$.  Suppose however, that $\deg_{\tilde G_1}(v_2) > 1$; in this case, $g_{\tilde G_1} - 1 = g_{\tilde G_2}$, since $V_{\tilde G_2} = V_{\tilde G_1}$.  Hence after 2 steps, we would find the number of choices \begin{equation}
   2 \tilde\Gamma^{1}_2 = 2\tilde\Gamma^{1}_1 \times \left(1 + \deg_{\tilde G_1}(v_2) y^{-1}\right) \le 2s \left(1+sy^{-1}\right).
\end{equation}
We define the inequality above by the property that every positive coefficient in the Laurent series is no larger on the left hand side than on the right.   Clearly, we can continue this process until we reach edge $L_1$.  Since after each step, the choices we make can be bounded by the same reasoning, summing over all possible pathways forward, we find that the number of choices we could make is \begin{equation}
    \tilde\Gamma^{1}_{L_1} \le s\left(1+sy^{-1}\right)^{L_1-1}.  \label{eq:tildeGamma1}
\end{equation}
Note that (\ref{eq:tildeGamma1}) is valid if $L_1=1$.

Now let us describe the segment $\mathcal{S}_2$.  Let us now define $\mathcal{S}_2^E = (e_1,e_2,\ldots)$.  By definition $e_1\notin \tilde G_{L_1}$, and $\tilde G_{L_1+1}$ will have one more edge.  If $V_{\tilde G_{L_1+1}} \ne V_{\tilde G_{L_1}}$, then we have added both a new edge and a new vertex to get to $\tilde G_{L_1+1}$, and so the net genus has not changed.  There are fewer than $N$ possible vertices to choose from.  If $V_{\tilde G_{L_1+1}} = V_{\tilde G_{L_1}}$, then adding edge $e_1$ increases the genus, since it does not add a new vertex, but adds a new edge.  There are fewer than $s+q$ vertices to choose from in the graph $\tilde G_{L_1}$ (this bound is not tight, but we will not need to adjust this bound at later steps of our ``algorithm").  Hence we conclude that after the first step of $\mathcal{S}^2$, the number of choices (in total) that we have made is \begin{equation}
    \tilde\Gamma^2_1 \le \tilde\Gamma^{1}_{L_1} \times \left(N + (s+q)y\right).
\end{equation}
As above, we can clearly repeat this process $L_2-2$ more times: 
\begin{equation}
    \tilde\Gamma^2_{L_2-1} \le \tilde\Gamma^{1}_{L_1} \times \left(N + (s+q)y\right)^{L_2-1}.
\end{equation}
However, at the last step, we have to be more careful.  If $\eta(C,G)=2$, then there is no freedom to choose the last edge.  If $\eta(C,G)>2$, then we can choose between at most $s+1$ vertices from $G$ to hit.  We conclude that \begin{equation}
    \tilde \Gamma^2_{L_2} \le \tilde\Gamma^{1}_{L_1} \times \left(N + (s+q)y\right)^{L_2-1} \times \left(1 + ((1+s)y-1)\mathbb{I}[\eta(C,G)>2]\right).
\end{equation}

Now, suppose that $\eta(C,G)>2$, so we must keep counting.  The counting for the third segment will be very similar to the first.  We have at most $s$ ways to move after the first step.
%\AO{I believe you are missing a factor of two here because of the first step. In equation 4.40 it looks like this is included but It is not clear to me whether or not it was actually included } \andy{We can ignore the first factor of 2 because we will count every cycle twice, reading it once in each direction.  Actually, we will read every cycle $\eta$ times, but I don't want to bother keeping track of that $\eta$, so I just said 2.}  
Subsequent intermediate steps are constrained as before.  At the last step, we know that we do not delete the last vertex since the next edge in $\mathcal{S}_4$ will include that vertex.  We conclude that \begin{equation}
    \tilde\Gamma^3_{L_3} \le \tilde \Gamma^2_{L_2} \times 2s \left(1+sy^{-1}\right)^{L_3-1} \times y^{-1}.
\end{equation}
During $\mathcal{S}_4$, the counting is essentially the same as $\mathcal{S}_2$: \begin{equation}
    \tilde \Gamma^4_{L_4} \le \tilde\Gamma^{3}_{L_3} \times \left(N + (s+q)y\right)^{L_4-1} \times \left(1 + ((1+s)y-1)\mathbb{I}(\eta(C,G)>4)\right).
\end{equation}
Clearly this accounting continues until we reach segment $\eta(C,G)$, at which case we arrive at our final bound: \begin{align} \label{eq:Gammay}
    \tilde \Gamma(y; L_1,\ldots ,L_\eta) &\le s(2s(1+s))^{\frac{1}{2}\eta-1} \prod_{n=1}^{\eta/2} \left(1+sy^{-1}\right)^{L_{2n-1}-1} \left(N + (s+q)y\right)^{L_{2n}-1} \notag \\
    &=  (2s(1+s))^{\frac{1}{2}\eta-1} \left(1+sy^{-1}\right)^{p-\frac{1}{2}\eta}\left(N + (s+q)y\right)^{q-p-\frac{1}{2}\eta} .
\end{align}
Note that \begin{equation}
   1\le \frac{\eta}{2} \le \min(p,q-p) := \eta_*.
\end{equation}
which comes from the fact that each segment $\mathcal{S}_i$ has at least one edge.  Also note that (\ref{eq:Gammay}) does not depend on $L_i$.  Therefore, \begin{align}
    \tilde\Gamma_\eta(y) &:= \sum_{\substack{L_1+L_3+\cdots = p \\ L_2+L_4+\cdots = q-p}} \tilde\Gamma_\eta(y,L_1,\ldots,L_\eta) \\
    &\le s(2p(q-p)s(1+s))^{\frac{1}{2}\eta-1} \left(1+sy^{-1}\right)^{p-\frac{1}{2}\eta}\left(N + (s+q)y\right)^{q-p-\frac{1}{2}\eta} \notag \\
    &\le s N^{q-p-1} \left(\frac{q^2s^2}{N}\right)^{\frac{1}{2}\eta-1} \left(1+\frac{s+q}{\sqrt{N}}\frac{\sqrt{N}}{y}\right)^{p-\frac{1}{2}\eta}\left(1+\frac{s+q}{\sqrt{N}} \frac{y}{\sqrt{N}}\right)^{q-p-\frac{1}{2}\eta}.
\end{align}
Using straightforward combinatorial bounds, we find that \begin{equation}
    \Gamma_{\eta, \delta g} \le s N^{q-p-1} \left(\frac{q^2s^2}{N}\right)^{\frac{1}{2}\eta-1} \left(1+\frac{s+q}{\sqrt{N}}\right)^{2(q-|\delta g|)} \left(\begin{array}{c} q \\ |\delta g| \end{array}\right)^2 \left(\frac{s+q}{\sqrt{N}}\right)^{|\delta g|} \times N^{-\frac{1}{2}\delta g},
\end{equation}
and evaluating the sum over $\eta$, we find \begin{equation}
    \Gamma_{\delta g} \le \sum_{\eta=2,4,\ldots}^\infty \Gamma_{\eta, \delta g} \le s N^{q-p-1} \frac{N}{N-q^2s^2}\left(1+\frac{s+q}{\sqrt{N}}\right)^{2q} \left(\frac{q^2(s+q)}{\sqrt{N}}\right)^{|\delta g|} \times N^{-\frac{1}{2}\delta g}.
\end{equation}
Assuming that $2q(s+q)<\epsilon\sqrt{N}$, we find \begin{equation}
    \Gamma_{\delta g} \le  \frac{4 \mathrm{e}^\epsilon}{4-\epsilon^2} \times s N^{q-p-1}  \left(\frac{q^2(s+q)}{\sqrt{N}}\right)^{|\delta g|} \times N^{-\frac{1}{2}\delta g}.
\end{equation}
\end{proof}

Propositions \ref{prop:CW} and \ref{prop:formalbound} and (\ref{eqn:boundind}) give the intuition for the remainder of  this proof. Proposition \ref{prop:CW} tells us that, despite the fact that $M_{s's}$ is neither irreducible nor positive, we may apply the Collatz--Weilandt formula to bound $\mu_{s's}$. 
   Proposition \ref{prop:formalbound} gives a bound on the number of cycles which may be chosen to commute through some graph to yield a graph with fixed genus and with this we can acquire a bound on (\ref{eqn:canbound}).  In order to bound $\mu_{s's}$, we
   need only to bound $\Lambda_{\delta g}$, the number of pairs of cycles $C$ and $C'$ so that 
   $(G \triangle C) \triangle C'$ has genus $g_G + \mathrm{\Delta} g$ and $s$ edges: \begin{equation}
       \Lambda_{\delta g}(G) := \lbrace C,C^\prime \in \mathcal{C}_q^N : G\triangle C \in \mathrm{K}_N^{s^\prime}\text{ and } (G\triangle C) \triangle C^\prime \in \mathrm{K}_N^s \text{ and } g_{(G\triangle C) \triangle C^\prime} = g+\Delta g \rbrace. \label{eq:Lambdadef}
   \end{equation} 
      Below, we will suppress the dependence on $G$ (and most importantly on $s$) when denoting $\Lambda_{\mathrm{\Delta} g}$.  Clearly, \begin{equation}
          \Lambda_{\delta g}(G) \le \sum_{n=-q}^q \Gamma_n(s\rightarrow s^\prime) \Gamma_{\mathrm{\Delta} g - n}(s^\prime \rightarrow s).
      \end{equation} 
      In the summand above, we have denoted $n=g_{G\triangle C} -g_G$ and $\mathrm{\Delta} g - n = g_{(G\triangle C)\triangle C^\prime} - g_{G\triangle C}$.  For clarity, we have explicitly emphasized in the notation above that the $\Gamma$ are evaluated with $s$ and $s^\prime$ switched in the two terms. In particular, fixing the values of $s$ and $p$, we can apply proposition \ref{prop:formalbound} to obtain 
   \begin{align}
     \Lambda_{\mathrm{\Delta} g} &\leq  \frac{16 e^{2\epsilon}}{(4 - \epsilon^2)^2} \times s \,s'\, N^{q-2} N^{-\frac{1}{2} \mathrm \Delta g}
       \sum_{n=-q}^q 
       \pfrac{q^2(s+q)}{\sqrt N}{|n|} \pfrac{q^2(s'+q)}{\sqrt N}{|\mathrm{\Delta}g-n|} \notag \\
       &< \frac{8 e^{2\epsilon}}{(4 - \epsilon^2)^2} \times (s+s^\prime)^2 N^{q-2} N^{-\frac{1}{2} \mathrm \Delta g}
       \sum_{n=-q}^q 
       \pfrac{q\epsilon}{2}{|n|+|\mathrm{\Delta}g-n|}. \label{eq:Lambdafinal}
   \end{align}
   where in the second line we have used (\ref{eq:propbounds}), and that since $s,s^\prime>0$,
   $s \, s' \leq \frac{1}{2}(s+s')^2$.
   
      From proposition \ref{prop:fermioncommute} and 
  (\ref{eqn:hamilt}), we see that for arbitrary $G, G' \subset \mathrm{K}_N$ and arbitrary $C,C' \in \mathcal{C}_q^N$, \begin{equation}\label{eqn:boundind}
     (G'|\QQ_s \LL_{C'}^\intercal \QQ_{s'} \LL_{C} \QQ_s |G) \leq 4 J_C J_{C'}
     \leq 4 \sigma^2 .
  \end{equation}
Now combining the definition of $M_{s^\prime s}$, (\ref{eq:Lambdadef}) and (\ref{eqn:boundind}),
   \begin{equation}
       \sum_{G^\prime : g_{G^\prime} = g+\mathrm{\Delta} g} (G^\prime |M_{s^\prime s}|G) \le 4\sigma^2 \Lambda_{\mathrm{\Delta} g}. \label{eq:Lambdasigma}
   \end{equation}
   Using the guess $|\phi)$ from Proposition \ref{prop:CW}, we combine (\ref{eqn:sigma}), (\ref{eq:Lambdafinal}), and (\ref{eq:Lambdasigma}) to obtain  \begin{equation}
       \mu_{s^\prime s} \le \frac{32 e^{2\epsilon}}{(4 - \epsilon^2)^2}  (s+s^\prime)^2 
       \sum_{\mathrm{\Delta}g=-2q}^{2q}  \sum_{n=-q}^q 
       \pfrac{q\epsilon}{2}{|n|+|\mathrm{\Delta}g-n|} \le \frac{32 e^{2\epsilon}}{(4 - \epsilon^2)^2}  (s+s^\prime)^2 \frac{4}{(2-q\epsilon)^2}.
   \end{equation}

Then from proposition \ref{prop:kdefn}, we see that 
\begin{equation}\label{eqn:}
  \mathcal{K}_{s's} \leq \sqrt 2  \frac{8e^\epsilon}{(4 - \epsilon^2)(2-q\epsilon)}    (s+s').
\end{equation}
Then, for $l$ obeying  $l \geq 0$, we use (\ref{eqn:walksuppos}) to find that whenever the bounds in (\ref{eq:thmbounds}) hold, \begin{align}
    \mathcal{K}_l &\le \sqrt{2} \frac{8e^\epsilon}{(4 - \epsilon^2)(2-q\epsilon)}  \sum_{n=2,4,\ldots}^{q-2} \left(2+2l(q-2) + n \right) \notag \\
    &= \sqrt{2} \frac{8e^\epsilon}{(4 - \epsilon^2)(2-q\epsilon)}(q-2)  \left[ 1+l(q-2)  + \frac{q}{4} \right].
\end{align}
Since $4+q \le 4(q-2)$ whenever $q\ge 4$, and when $q=2$ this expression vanishes, we obtain (\ref{eqn:toprove}).
\end{proof}

Using Theorem \ref{thm:sykthm}, for $N$ sufficiently large, we can obtain a constant $\alpha$ arbitrarily close to 0 in (\ref{eq:thm35alpha}), and hence we conclude that  the constant $c$ in (\ref{eq:thm35c}) obeys 
  \begin{equation}
      c = \sqrt 2 \frac{8 e^\epsilon}{(4 - \epsilon^2)(2 - q \epsilon)}(q-2)^2.
  \end{equation}
  For arbitrarily fixed (i.e. $N$--independent) $q$, we may choose $\epsilon = N^{-\frac{1}{4}}$ and take the limit 
  \begin{equation}
      \lim_{N \rightarrow \infty}c =\sqrt 2 (q - 2)^2.
  \end{equation}
  In this case, the hypothesis of Proposition \ref{prop:formalbound} is satisfied and the application of (\ref{eq:thm35alpha}) from theorem \ref{thm:sykthm} straightforwardly produces (\ref{eq:lambdaLboundloose}). 
  On the other hand, as our model is defined, theorem \ref{thm:sykthm} does not necessarily hold if $q$ is $N$--dependent. If, however, we had defined $\sigma$ to include a factor of $q^{-2}$, our bound would be robust in the $N \rightarrow \infty$ limit with scaling as severe as $q \propto N^{1/6 - \beta}$ for arbitrary $\beta > 0$. In this case we would instead have that \begin{equation}
      \lim_{N\rightarrow \infty}c \le \sqrt 2
  \end{equation}
  and in place of (\ref{eq:lambdaLboundloose}) we would have 
  \begin{equation}
      \lambda_{\mathrm L}\le 2 \sqrt 2.
  \end{equation}
  The scaling of $q\lesssim N^{1/6-\beta}$ is necessary so that both bounds in (\ref{eq:thmbounds}) can hold for $l\gg 1$ for large enough $N$.
  
  While we will not present a Feynman diagrammatic analysis of this model here, we briefly note that at leading order in the large $N$ expansion, the class of operators $|G)$ that arise in a growing operator $|e(t))$ (where $e\in E_{\mathrm{K}_N}$ and $e=\lbrace u,v\rbrace$) is quite restrictive: $G$ must correspond to a genus zero graph with two degree one vertices ($u$ and $v$), and all others degree two.  Such graphs correspond to ``lines" (alternatively: ``chains" or ``worms") whose endpoints are $u$ and $v$.  This fact follows in part from Proposition \ref{prop:parity}, which shows that the $\mathbb{Z}_2$-valued degree of every vertex (i.e., is the degree even or odd?) obeys: \begin{equation}
      \deg_G(w; \mathbb{Z}_2) = \deg_H(w; \mathbb{Z}_2),\;\;\; \text{for all } w\in V, \;\;\; \text{if } \;\;\; (H|G(t)) \ne 0.
  \end{equation}   We conjecture that a similar pattern of operator growth holds  in the large $N$ limit of all matrix models, including those with bosonic degrees of freedom.
  
\subsection{Comparison to the Sachdev-Ye-Kitaev model}
We conclude the paper with a few technical (but not rigorous) remarks about our result, in light of earlier work \cite{Lucas:2019cxr} which derived similar bounds for the SYK model.

One obvious difference between the SYK model and the cartoon matrix model is that the former relies heavily on randomness in order to have universal patterns of operator growth, whereas the latter does not appear to.\footnote{Of course, it may be the case that without randomness in the cartoon matrix model, there is severely destructive quantum interference which sends the Lyapunov exponent to zero.  However, we doubt this effect will be present.}  However, the role of randomness may be relatively minor in the broader picture:  other melonic models \cite{Gurau:2011xp,Gurau:2016lzk,Witten:2016iux,Klebanov:2016xxf,Gubser:2018yec} are not random.  

Most likely, the main qualitative difference between SYK operator growth and matrix models more generally (when $q>4$) is that in matrix models, operators need not grow in a fixed sequences of sizes, as they do in SYK \cite{Roberts:2018mnp}: $1\rightarrow q-1\rightarrow 2q-3\rightarrow \cdots$ in the large $N$ limit.  In our matrix model, we were unable to rule out substantial contributions to growth from operators adding  $q-2$, $q-4$, etc. fermions at a time.  To see the consequences of this, it is useful to compare a little more explicitly the $q$-dependence in $\lambda_{\mathrm{L}}$ between our matrix model and SYK.  In the SYK model of $N$ fermions, one finds that \begin{equation}
   \frac{ \mathrm{tr}\left(H^2\right)}{\mathrm{tr}(1)} \sim \frac{N}{q^2}, \;\;\; \text{ implies }  \;\;\; \lambda_{\mathrm{L}} \sim q^0,
\end{equation}
while in the matrix model of $\sim N^2$ fermions, one finds that, using (\ref{eqn:sigma}) and estimating $|\mathcal{C}_q^N| \sim q^{-1}N^q$, \begin{equation}
    \frac{ \mathrm{tr}\left(H^2\right)}{\mathrm{tr}(1)} \sim \frac{N^2}{q} , \;\;\; \text{ implies }  \;\;\; \lambda_{\mathrm{L}} \sim q^2.
\end{equation}
The above equations imply that had we scaled the coupling constant similar to SYK, we would have found $\lambda_{\mathrm{L}} \sim q^{3/2}$ at large $q$.  We can understand this discrepancy as follows: one factor of $q$ arises from the fact that $\mathcal{K}_{s+q-2,s}$ and $\mathcal{L}_{s+q-4,s}$ are comparable in size (in contrast to the SYK model, where the latter is suppressed by powers of $N$ \cite{Lucas:2019cxr}).  A factor of $q^{1/2}$ arises from the fact that there are $\sim s$ couplings to choose from when reducing an operator in size from $s+q-2\rightarrow s$, even at leading order in thhe large $N$ limit.  This was not true in the SYK model: due to the randomness in the model, there were effectively only $\sim s/q$ ways to reduce the size of an operator of size $s$ (that would not, after disorder averaging, lead to subleading effects in $1/N$).  However, since the SYK model is dominated by the same diagrams as other melonic models, we expect that other melonic models are ``frustrated" enough that there are only $\sim s/q$ ways to reduce an operator of size $s$.

A pattern that is shared by our cartoon matrix model and the SYK model is the fact that of all the operators of size $s$, the ones which can grow the fastest are those which themselves grew out of a size $\sim 1$ operator \cite{Lucas:2019cxr}.  In the cartoon matrix model, it is easy to see why ``generic" operators do not grow quickly.  A randomly chosen subset of $s$ edges will consist of a completely disconnected graph until $s\sim \sqrt{N}$; until $s\sim N$, the graph will consist of tiny disconnected fragments of size $\lesssim \log N$ \cite{bollobas}.  The growth rates of such graphs $G$ would scale as $ \lVert \mathbb{Q}_{s+q-2}\mathcal{L}\mathbb{Q}_s|G) \rVert_2 \sim \sqrt{s}$, rather than $s$.  It is only the highly connected graphs (such as the ``chains" discussed above) which can come close to saturating the bound (\ref{eqn:toprove}).   It would be interesting to understand, in complete generality, which quantum many-body systems have this pattern of operator growth. It may be a sensible route to looking for experimentally simulatable models of holographic two-dimensional quantum gravity.

\section*{Acknowledgements}
AL is supported by a Research Fellowship from the Alfred P. Sloan Foundation.

\bibliography{thebib}

\end{document}